\documentclass[sn-mathphys]{sn-jnl}

\theoremstyle{thmstyleone}%
\newtheorem{theorem}{Theorem}
\newtheorem{lemma}[theorem]{Lemma}
\theoremstyle{thmstyletwo}%

\theoremstyle{thmstylethree}%
\newtheorem{definition}{Definition}%

\usepackage{mathtools}

\renewcommand{\R}{\mathbb{R}}
\newcommand{\Rn}[1][n]{\mathbb{R}^{#1}}

\newcommand{\K}{\mathcal{K}}
\newcommand{\Kmin}{\mathcal{K}_{\min}}

\newcommand{\pgame}{(N,\K,v)}
\newcommand{\minpgame}{(N,\Kmin,v)}

\renewcommand{\S}{S^n}
\newcommand{\Sv}{S^n(v)}
\newcommand{\Co}{C^n}

\renewcommand{\P}{P^n}
\newcommand{\Pv}{P^n(v)}
\newcommand{\oneC}{C_1^n}
\newcommand{\oneCv}{C_1^n(v)}

\newcommand{\Core}{\mathcal{C}}
\newcommand{\I}{\mathcal{I}}

\newcommand{\W}{\mathcal{W}}

\newcommand{\walfa}{{w_{\alpha}}}
\newcommand{\walfabeta}{{w_{\alpha,\beta}}}
\newcommand{\wa}{{w_{A}}}

\newcommand{\marg}[1]{m^{#1}_\sigma}
\newcommand{\margi}[2]{\left(m^{#1}_\sigma\right)_{#2}}

\begin{document}

\title[Approximations of solution concepts of cooperative games]{Approximations of solution concepts of cooperative games}

\author{\fnm{Martin} \sur{\v{C}ern\'{y}}}\email{cerny@kam.mff.cuni.cz}

\affil{\orgdiv{Department of Applied Mathematics}, \orgname{Charles University}, \orgaddress{\city{Prague}, \country{Czech Republic}}}


\abstract{The computation of a solution concept of a cooperative game usually depends on values of all coalitions. However, in some applications, values of some of the coalitions might be unknown due to various reasons. We introduce a method to approximate standard solution concepts based only on partial information given by a so called incomplete game.
	
We demonstrate the ideas on the class of minimal incomplete games. Approximations are derived for different solution concepts including the Shapley value, the nucleolus, or the core. We show explicit formulas for approximations of some of the solution concepts and show how the approximability differs based on additional information about the game.}





\maketitle

\section{Introduction}
The model of cooperative games is long studied since its original formulation by von Neumann and Morgenstern in 1944~\cite{Neumann1944}. Since then, one of its main disadvantages, which makes it often infeasible in applications prevails. If we want to model a real world problem using a cooperative game of $n$ players, we need to collect $2^n$ real values representing the worth of cooperation between different subsets of players. In the best case scenario, this process is both resource and time demanding but in many scenarios, such information is too complex and unfeasible to collect. Researchers try to deal with these problems in many ways; using stochastic approach~\cite{Granot1977,Borm2002}, approach employing (fuzzy) intervals~\cite{Gok2009,Bok2015,Branzei2010,Deng2016,Mares2013}, among other approaches are e.g. ellipsoidal games~\cite{Gok2014}, multi-choice games~\cite{Branzei2008}, or restricted games~\cite{Grabisch2016}. When computing the solution concepts, in many of these models the unknown values are either ignored, their values are substituted with $0$ or the solution concept is computed with respect to a special cooperative game. This is hardly satisfactory, because the unknown values are either not taken into account or there is no guarantee on their relation with the substituted values.

In this text, we derive approximations of standard solution concepts, given only partial information about an underlying cooperative game. Section~\ref{sec2} contains preliminaries to cooperative games and section~\ref{sec3} is dedicated to the model of incomplete cooperative games that we use as the tool to capture partial information about an underlying complete game. Section~\ref{sec4} is dedicated to approximations of standard solution concepts and in Section~\ref{sec5}, we conclude this paper with open problems and future research.

\section{Complete cooperative games}\label{sec2}
This section presents only the necessary background. We define the cooperative game, different classes of games and solution concepts studied in our research. For more on cooperative games, see~\cite{Grabisch2016,Peleg2007,Peters2008}. 
\begin{definition}
	A \emph{cooperative game} is an ordered pair $(N,v)$, where $N=\{1,2,\ldots,n\}$ and $v\colon 2^N \to
	\mathbb{R}$ is the characteristic function. Further, $v(\emptyset) = 0$.
\end{definition}

We denote the set of $n$-person cooperative games by $\Gamma^n$. Subsets of $N$ are called \emph{coalitions}, $\{i\}$ for $i \in N$ are \emph{singletons} and $N$ is called the
\emph{grand coalition}. We often write $v$ instead of $(N,v)$ when there
is no confusion over the player set. We associate the characteristic functions $v\colon 2^N\to\mathbb{R}$ with vectors $v\in\mathbb{R}^{2^{\lvert N \rvert}}$. This is convenient for viewing sets of cooperative games as sets of points.
We use the following abbreviations. We often replace $\left\{i\right\}$ with $i$. To denote the sizes of coalitions e.g. $N,S,T$, we use $n,s,t$. For $x \in \Rn$ and $S \subseteq N$, $x(S) \coloneqq \sum_{i \in S} x_i$. By $\mathbb{R}_+$, we denote the set of non-negative real values.

\begin{definition}\label{def:classes-of-games}
	Let $(N,v)$ be a cooperative game. The game is
	\begin{enumerate}
		\item {\it superadditive} if it satisfies
		\begin{equation*}
			v(S) + v(T) \le v(T\cup S), \hspace{5ex} S, T \subseteq N, S\cap T=\emptyset;
		\end{equation*}
		\item {\it convex} if it satisfies
		\begin{equation*}
			v(S) + v(T) \le v(T\cup S) + v(T\cap S), \hspace{5ex} S, T \subseteq N.
		\end{equation*}
	\end{enumerate}
	We denote the sets of superadditive and convex $n$-person games by $\S$ and $\Co$.
\end{definition}
\emph{Unanimity games} $(N,u_T)$ for nonempty $T \subseteq N$ defined as
\[
u_T(S) \coloneqq 
\begin{cases}
	1 & \text{if }T \subseteq S,\\
	0 & \text{otherwise},\\
\end{cases}
\]
are important when we view the set of all $n$-person cooperative games $\Gamma^n$ as a vector space. Shapley~\cite{Shapley1953} showed that unanimity games form one of its bases, i.e. every game $v \in \Gamma^n$ can be expressed as $v= \sum_{T \subseteq N,T \neq \emptyset}d_v(T)u_T$, or equivalently $v(S) = \sum_{T \subseteq S, T \neq \emptyset}d_v(T)$. The coefficients of this linear combination $d_v(T)$ are called \emph{Harsanyi dividends} and can be expressed as
\[
d_v(T) = \sum_{S \subseteq T}(-1)^{\lvert T \rvert - \lvert S \rvert}v(S).
\]
\begin{definition}
	A cooperative game $(N,v)$ is \emph{positive}, if it holds for all coalitions $\emptyset \neq T \subseteq N$ that \[d_v(T) \geq 0.\]
	We denote the set of all positive cooperative $n$-person games by $\P$.
\end{definition}
It is straightforward that convex games form a subset of superadditive games. Further, positive games form a subset of convex games and unanimity games are themselves positive games. Positive games are well-studied in the theory of capacities, where they are also called \emph{totally-monotonic games} (see~\cite{Grabisch2016}).

The \emph{upper vector} $b^v \in \mathbb{R}^n$ defined as $b^v_i\coloneqq  v(N) - v(N\setminus i)$ captures each player's marginal contribution to the grand coalition. The \emph{lower vector} $a^v$ is defined as $a^v_i = \max_{S: i \in S} v(S) - b^v(S)$. Finally, the \emph{gap function} $g^v \colon 2^N \to \R$ is defined as $g^v(S) \coloneqq b^v(S) - v(S)$. The following lemma shows that the upper vector and the gap function are linear with respect to cooperative games.

\begin{lemma}\cite{Driessen1985-thesis}\label{lem:linear-combination}
	For a linear combination $v = \sum_{i=1}^k\alpha_i v_i$ of games $v_1,\dots,v_k$, it holds
	\begin{enumerate}
		\item $b^v = \sum_{i=1}^k\alpha_ib^{v_i},$
		\item $b^v(N)=\sum_{i=1}^k\alpha_ib^{v_i}(N),$
		\item $g^v(N)=\sum_{i=1}^k\alpha_ig^{v_i}(N).$
	\end{enumerate}
\end{lemma}

\begin{definition}\label{def:oneC-games}
	Let $(N,v)$ be a cooperative game. The game is 1-convex if for all coalitions $S \subseteq N$, $S \neq \emptyset$ the inequality
	\begin{equation}\label{def:1conv-cond1}
		v(S) \leq v(N) - b^v(N\setminus S)
	\end{equation}
	holds and also
	\begin{equation}\label{def:1conv-cond2}
		b^v(N) \geq v(N).
	\end{equation}
	The set of 1-convex $n$-person games is denoted by $\oneC$.
\end{definition}
\subsection{Solution concepts and payoff vectors}
One of the goals of cooperative game theory is to distribute the value of the grand coalition $v(N)$ between all players. To be able to work with individual payoffs more easily, \emph{payoff vectors} are introduced. Those are vectors $x \in \Rn$ where $x_i$ represents the individual payoff of player $i$. The definition of the payoff vector is quite general. This is why, for a cooperative game $(N,v)$, one usually considers \emph{preimputations} $\I^*(v)$ and \emph{imputations} $\I(v)$,
\begin{itemize}
	\item $\I^*(v) \coloneqq \{x \in \Rn \mid x(N) = v(N) \}$,
	\item $\I(v) \coloneqq \{ x \in \Rn \mid x(N) = v(N) \text{ and }  x_i \geq v(i) \text{ for } i \in N\}$.
\end{itemize}
Preimputations are payoff vectors $x \in \Rn$ that are \emph{efficient}, i.e. $x(N) = v(N)$. Imputations are also \emph{individually rational}, meaning every player $i$ receives at least his singleton value, i.e.\, $x_i \geq v(i)$. To work with imputations more easily, for $k \in N$, we denote by $I^k \in \Rn$ the imputation defined as
\[
(I^k)_i = \begin{cases}
	v(i) + \delta & \text{if } i =k,\\
	v(i) & \text{if } i \neq k.\\
\end{cases}
\]
and for $\alpha \in \Rn_+$ such that $\alpha(N)=1$, we denote by $I^\alpha \in \Rn$ the imputation defined as
\[
I^\alpha_i = v(i) + \alpha_i\Delta.
\]
\begin{definition}
	Let $C\subseteq \Gamma^n$ be a class of $n$-person cooperative games. Then a function $f\colon C\to2^{\Rn}$ is a \emph{solution concept} (on class $C$).
\end{definition} 
Solution concepts might be equivalently defined as subsets of payoff vectors. Both approaches are useful and we often switch between them. If the image $f(v)$ of every cooperative game $v \in C$ is exactly one vector, we write $f\colon C\to\Rn$ and we say $f$ is a \emph{one-point} solution concept or a \emph{value}. Otherwise, we say $f$ is a \emph{multi-point} solution concept. Each solution concept follows a different goal, e.g.\, the \emph{Shapley value} (Definition~\ref{def:shapley}) is a one-point solution concept, which strives to distribute the payoff as \emph{fairly} as possible. Another example is the \emph{core} (Definition~\ref{def:core}), a multi-point solution concept, focused on \emph{stability}. In the rest of this section, we introduce solution concepts considered in our research.

\subsubsection{The Shapley value}
It is one of the most studied one-point solution concept originally defined by Shapley~\cite{Shapley1953}. Here, we use an alternative characterisation from~\cite{Peters2008}.
\begin{definition}\label{def:shapley}
	The \emph{Shapley value} $\phi \colon \Gamma^n \to \mathbb{R}^n$ is a one-point solution concept defined as
	\[\phi_i(v) \coloneqq \frac{1}{n}\sum\limits_{S\subseteq N \setminus i}{n-1\choose s}^{-1}(v(S\cup i) - v(S)).\]
\end{definition}
The Shapley value is defined for every cooperative game and has many nice properties, e.g. the \emph{linearity}~\cite{Shapley1953},
\begin{equation}\label{eq:shapley-linearity}
	\phi(\alpha v+\beta w) = \alpha\phi(v) + \beta\phi(w).
\end{equation}
In general, the Shapley value is not individually rational, i.e.\, $\phi(v) \notin \mathcal{I}(v)$. Further in this text, we will work with Shapley values of the unanimity games.
\begin{lemma}\cite{Shapley1953}\label{lem:shapley-unanimity}
	The Shapley value of a unanimity game $(N,u_T)$ can be expressed as
	\begin{equation*}
		\phi_i(u_T) = \begin{cases}
			\frac{1}{\lvert T \rvert} & \text{if } i \in T,\\
			0 & \text{if } i \notin T.\\    
		\end{cases}
	\end{equation*}
\end{lemma}

\subsubsection{The $\tau$-value}
The $\tau$-value is a one-point solution concept originally defined for the class of \textit{quasi-balanced} games~\cite{Driessen1985-thesis}, which is a superset of both of convex and 1-convex games. One of its equivalent forms views it as the efficient compromise between the upper vector $b^v$ and the lower vector $a^v$.
\begin{definition}
	The $\tau$-value $\tau(v)$ of a cooperative game $(N,v)$ is the unique convex combination of $a^v$ and $b^v$ satisfying $\sum_{i \in N}\tau_i(v)=v(N)$.
\end{definition}
For both convex games and 1-convex games, we have explicit formulas for the $\tau$-value.
\begin{theorem}\cite{Driessen1985-thesis}\label{thm:convex-tau}
	For a convex $(N,v)$, the $\tau$-value can be expressed as
	\begin{equation*}
		\tau_i(v) = \begin{cases}
			b^v_i - \frac{g^v(N)}{\sum_{i \in N}g^v(i)}g^v(i) & \text{if } g^v(N)=0,\\
			b^v_i & \text{if } g^v(N) > 0.\\
		\end{cases}
	\end{equation*}
\end{theorem}

\begin{theorem}\cite{Driessen1985-thesis}\label{thm:1convex-tau}
	For a 1-convex $(N,v)$, the $\tau$-value can be expressed as
	\begin{equation*}
		\tau_i(v) = b^v_i - \frac{g^v(N)}{n}.
	\end{equation*}
\end{theorem}

The $\tau$-value is always contained in the imputation set, i.e. $\tau(v) \in \I(v)$~\cite{Driessen1985-thesis}.

\subsubsection{The core and the Weber set}
\begin{definition}\label{def:core}
	The \emph{core} $\Core(v)$ of a cooperative game $(N,v)$ is a multi-point solution concept defined as
	\[\Core(v) \coloneqq \{x \in \Rn \mid x(N)=v(N) \text{ and } x(S) \geq v(S) \text{ for } S \subseteq N\}.\]
\end{definition}
Notice, $\Core(v) \subseteq \I(v)$ since $x(N)=v(N)$ and for $S=\{i\}$, we have $x_i \geq v(i)$. The core satisfies $x(S) \geq v(S)$ for $S \subseteq N$, a property called \emph{coalitional rationality}. The core of a superadditive game might be empty, but for convex and 1-convex games it is always nonempty. For 1-convex game, it can be expressed as a convex hull of easily computable vectors.
\begin{theorem}\label{thm:1convex-core}\cite{Driessen1985}
	For 1-convex game $(N,v)$, it holds
	\[ \Core(v) = conv\left\{b^v-g^v(N)e_i \mid i \in N\right\}.\]
\end{theorem}

The Weber set is a multi-point generalisation of the Shapley value. It is a \emph{core-catcher}, meaning it always contains the core and it coincides with the core if and only if the underlying game is convex.
For a permutation $\sigma \in \Sigma_n$, the set of predecessors of $i$ with respect to $\sigma$ is $S_{\sigma(i)} \coloneqq \{ j \in N \mid \sigma(j) < \sigma(i)\} $. A \emph{marginal vector} $\marg{v} \in \Rn$ is then defined as
\begin{equation}\label{eq:marg-vector-definition}
	\margi{v}{i} = v(S_{\sigma(i)} \cup i) -v(S_{\sigma(i)}).
\end{equation}
\begin{definition}\label{def:weber}
	The \emph{Weber set} $\W(v)$ of a cooperative game $(N,v)$ is defined as 
	\[
	\W(v) \coloneqq conv\{\marg{v} \mid \sigma \in \Sigma_n\}.
	\]
\end{definition}

\begin{theorem}\label{thm:convex-weber-equals-core}\cite{Peleg2007}
	For every cooperative game $(N,v)$, it holds $\Core(v)\subseteq \W(v)$. Further, $\Core(v)=\W(v)$ if and only if $(N,v)$ is convex.
\end{theorem}

\subsubsection{The (pre)kernel}
The \emph{excess} of $(N,v)$ with respect to $x \in \Rn$ is $e(S,x,v) \coloneqq v(S) - x(S)$ and the \emph{maximal surplus of $i$ over $j$ at $x \in \Rn$} is
\[
s_{ij}(x,v) \coloneqq \max_{S: i \in S, j \notin S}e(S,x,v).
\]
\begin{definition}\label{def:prekernel}
	The \emph{prekernel} $\K^*(v)$ of a cooperative game $(N,v)$ is a multi-point solution concept defined as
	\[
	\K^*(v) \coloneqq \{x \in \I^*(v) \mid s_{ij}(x,v) = s_{ji}(x,v) \text{ } \forall i \neq j \}.
	\]
\end{definition}

The definition of the kernel is slightly more restrictive than that of the prekernel. The main difference is in individual rationality of the payoff vectors.
\begin{definition}\label{def:kernel}
	The \emph{kernel} $\K(v)$ of a cooperative game $(N,v)$ is a multi-point solution concept defined as
	\begin{align*}
		\K(v) = \Big\{ x \in \I(v) \mid \forall i\neq j: \left(s_{ij}(x,v) - s_{ji}(x,v)\right)(x_j - v_j) \leq 0 \text{ \,} \\
		\text{ or } \left(s_{ij}(x,v) - s_{ij}(x,v)\right)(x_i - v_i) \leq 0\Big\}&.
	\end{align*}

\end{definition}
Either $s_{ij}(x,v) = s_{ji}(x,v)$ or, from individually rational, $s_{ij(x,v)} > s_{ji}(x,v)$ implies $x_j=v_j$ and similarly $s_{ji}(x,v) > s_{ij}(x,v)$ implies $x_i=v_i$. Although different in general, both solution concepts coincide for superadditive games.
\begin{theorem}\label{thm:kernel-equals-prekernel}\cite{Peleg2007}
	For every superadditive game $(N,v)$, it holds $\K^*(v) = \K(v)$.
\end{theorem}

\subsubsection{The nucleolus}
The essential component of nucleolus is $\theta(x) \in \mathbb{R}^{2^n}$, the vector of excesses with respect to $x$ which is arranged in non-increasing order.

\begin{definition}
	The \emph{nucleolus}, $\eta\colon \Gamma^{n}\to \mathbb{R}^n$ of a cooperative game $(N,v)$ is a one-point solution concept which assigns to $(N,v)$ the minimal imputation $x$ with respect to the lexicographical ordering $\theta(x)$ defined as:
	\[
	\theta(x) < \theta(y) \text{ if } \exists k: \forall i < k: \theta_i(x)=\theta_i(y) \text{ and } \theta_k(x) < \theta_k(y).
	\]
\end{definition}

The nucleolus is a \emph{core selector}, meaning $\eta(n) \in \Core(v)$ whenever the core is nonempty~\cite{Schmeidler1969}. For the class of convex games and 1-convex games, the nucleolus coincides with different solution concepts.

\begin{theorem}\cite{Maschler1971}\label{thm:convex-nucleolus-coincidence}
	Let $(N,v)$ be a convex game. Then $\eta(v)=\K(v)=\K^*(v)$.
\end{theorem}

\begin{theorem}\cite{Driessen1985-thesis}\label{thm:oneC-nucleolus-equals-tau}
	Let $(N,v)$ be a 1-convex game. Then $\eta(v)=\tau(v)$.
\end{theorem}

\section{Incomplete cooperative games}\label{sec3}
\begin{definition}\emph{(Incomplete game)}\label{def:incomplete-game}
	An incomplete game is a tuple $(N,\K,v)$ where $N=\{1,\dots,n\}$, $\K \subseteq 2^N$ is the set of coalitions with known values and $v\colon 2^N \to \R$ is the characteristic function of the incomplete game. Further, $\emptyset \in \K$ and $v(\emptyset)=0$.
\end{definition}

For \emph{minimal incomplete games}, we have $\K=\{\emptyset,N\} \cup \{\{i\} \mid i \in N\}$. This class is considered minimal, because we want $N \in \K$ to be able to distribute $v(N)$ among the players. Combined with that, the knowledge of singleton values allows to define the imputation set $\mathcal{I}(v)$ of $(N,\K,v)$.

Definition~\ref{def:incomplete-game} is actually identical to the definition of a \emph{restricted game} (see~\cite{Grabisch2016}). The crucial difference between both models is in the interpretation of the set $\K$. For restricted games, $\K$ represents all feasible coalitions and the cooperation of coalitions outside $\K$ is impossible. Thus the values are considered non-existent and are not taken into account while computing solution concepts. For incomplete games, though, $\K$ is the set of coalitions with known values. This means that we consider there is an underlying complete game $(N,v)$, however, all we know about this game is represented by $(N,\K,v)$. It is further assumed we know the underlying game $(N,v)$ is from a class $C \subseteq \Gamma^n$. If we want to approximate the solution concept of $(N,v)$, we can utilize this knowledge and consider all games from $C$, which are extending $(N,\K,v)$.

\begin{definition}\emph{($C$-extension)}
	Let $C\subseteq \Gamma^n$ be a class of $n$-person games. A cooperative game $(N,w) \in C$ is a \emph{$C$-extension} of an incomplete game $(N,\K,v)$, if $w(S)=v(S)$ for every $S \in \mathcal K$. 
\end{definition}

The set of all $C$-extensions of an incomplete game $(N,\K,v)$ is denoted by $C(v)$. We write \emph{$C(v)$-extension} whenever we want to emphasize the game $(N,\K,v)$. In some situations, we might not know the class of the underlying game. In this case, it even makes sense to ask if it has a $C$-extension. If it does, we say $(N,\K,v)$ is \emph{$C$-extendable}. The set of all $C$-extendable incomplete games with fixed $\K$ is denoted by $C(\K)$.

The sets of $\S$-extensions, $\Co$-extensions, $\P$-extensions and $\oneC$-ex\-ten\-sions form polyhedral sets, because they are defined by systems of linear inequalities. This allows to switch to the dual description of these sets, employing their extreme points and extreme rays. We refer to the extreme points as to \emph{extreme games}. To describe the extreme points and rays explicitly is often a challenging task. In the rest of this section, we discuss known results for minimal incomplete games.

\subsection{Sets of $C$-extensions}
The \emph{total excess} $\Delta$ is defined as $\Delta \coloneqq v(N) - \sum_{i \in N}v(i)$. We initially turn our focus to $\P$-extensions.

\begin{theorem}\cite{Bok2020}
	Let $(N,\K,v)$ be a minimal incomplete game with $v(i) \geq 0$ for all $i \in N$. It is $\P$-extendable if and only if $\Delta \geq 0$.
\end{theorem}

In~\cite{Masuya2016}, Masuya and Inuiguchi described extreme games of the set of $\P$-ex\-ten\-sions. Those are games $(N,v_T)$, parameterised by coalitions from set $N_1\coloneqq \{T \subseteq N \mid \lvert T \rvert > 1\}$. The games are expressed as
\begin{equation}\label{eq:def-vt}
	v_T(S) = \begin{cases}
		0, &  S = \emptyset,\\
		\Delta + \sum_{i \in S}v(i), &  S \notin\K \text{ and } T \subseteq S,\\
		\sum_{i \in S}v(i), &  S \notin\K\text{ and } T \subsetneq S.\\
	\end{cases}
\end{equation} 
\begin{theorem}\label{thm:positive-extensions-min-info}\cite{Masuya2016}
	For a $\P$-extendable minimal incomplete game $\pgame$, the set of $\P$-ex\-ten\-sions can be expressed as
	\begin{equation}\label{eq:positive-extensions}
		\Pv=\left\{\sum_{T \in N_1}\alpha_Tv_T \mid \sum_{T \in N_1}\alpha_T=1, \alpha_{T} \geq 0\right\},
	\end{equation}m
	where $(N,v_T)$ for $T \in N_1$ are games from (\ref{eq:def-vt}). 
\end{theorem}
We denote by $(N,\wa)$ a $\P$-extension where $A = (\alpha_T)_{T \in N_1}$, such that 
\begin{equation}\label{eq:positive-wa}
	\wa = \sum_{T \in N_1}\alpha_Tv_T.
\end{equation}

For the sets of $\Co$-extensions and $\S$-extensions, only partial results are known. Clearly, thanks to the relation of the classes, it holds 
\begin{equation}\label{eq:relation-of-extensions}
	\Pv \subseteq \Co(v) \subseteq \Sv.
\end{equation} 
Masuya and Inuiguchi~\cite{Masuya2016} showed that games $(N,v_T)$ for $T \in N_1$ are extreme points for all of these sets, however, it can be showed there are other extreme games for the latter two sets, e.g. games $(N,v^k)$ for $k \in N$,
\begin{equation}\label{eq:oneC-extreme-games}
	v^k(S) \coloneqq \begin{cases}
		v(i) & S = \{i\}, \\
		\Delta + \sum_{j \in S} v(j) & S \neq \{i\} \text{ and } k \in S, \\
		\sum_{j \in S} v(j)  & S \neq \{i\} \text{ and } k \notin S. \\
	\end{cases}
\end{equation}
For both sets, there are other extreme games apart from those already mentioned. The main difficulty in analysis of both sets of $\Co$-extensions and $\S$-extenisons lies in the expression of such extreme games. We omit the proof of extremality as it is not necessary for our analysis.
\begin{lemma}\label{lem:extreme-games}
	Let $(N,\K,v)$ be a $\Co$-extendable minimal incomplete game. Then games $(N,v^k)$ for $k \in N$ defined in~\eqref{eq:oneC-extreme-games} are convex.
\end{lemma}
\begin{proof}
	Immediate from Definition~\ref{def:classes-of-games} of convex games and the definition of $(N,v^k)$.
\end{proof}

The set of $\oneC$-extensions differs from the already mentioned sets because it is not only fully defined by extreme points, but also by extreme rays. Fortunately, all these games are known to us.
\begin{theorem}\cite{Bok2021}\label{thm:excess-extending}
	An incomplete game $\minpgame$ is $\oneC$-extendable if and only if $\Delta \geq 0$.
\end{theorem}
The extreme games of the set of $\oneC$-extensions are exactly games $(N,v^k)$ for $k \in N$ from~\eqref{eq:oneC-extreme-games}. The extreme rays are games $(N,e_T)$ for $T \in E$ where $E = 2^N \setminus \left(\{0,N\} \cup \{N\setminus i, \{i\} \mid i \in N \}\right)$,
\begin{equation}\label{eq:oneC-extreme-rays}
	e_T(S) \coloneqq \begin{cases}
		-1, & \text{if } S = T,\\
		0, & \text{if } S \neq T.\\
	\end{cases}
\end{equation}
\begin{theorem}\label{thm:oneC-min-info-set}\cite{Bok2021}
	For a $\oneC$-extendable minimal incomplete game $(N,\K,v)$, the set of $\oneC$-ex\-ten\-sions can be described as
	\[
	\oneCv = \left\{\sum_{i \in N}\alpha_i v^i + \sum_{ T \in E}\beta_Te_T \mid \sum_{i \in N}\alpha_i = 1 \text{ and } \alpha_i,\beta_T \geq 0\right\}.
	\]
\end{theorem}
To talk about different $\oneC$-extensions, we use the following notation. We denote by $(N,w_\alpha)$ where $\alpha \in \mathbb{R}^n_+$ and $\alpha(N) =1$ the convex combinations of extreme games $(N,v^i)$ for $i \in N$, i.e.\ $w_\alpha = \sum_{i \in N}\alpha_iv^i$. We can express $(N,\walfa)$ equivalently as
\begin{equation}\label{eq:oneC-walfa}
	w_\alpha(S) = \begin{cases}
		v(S) & S \in \K,\\
		\sum_{j \in S}\left(v(j) + \alpha_j\Delta\right) & S \notin \K.\\
	\end{cases}
\end{equation}
To denote a general $\oneC$-extension in a similar manner, we use $(N,\walfabeta)$ where $\alpha \in \Rn_+$, $\alpha(N) = 1$ and $\beta = (\beta_T)_{T \in E} \geq 0$. Using these, we can now express $w_{\alpha,\beta} = \sum_{i \in N}\alpha_i v^i + \sum_{T \in E}\beta_Te_T$ or equivalently
\begin{equation}\label{eq:oneC-walfabeta}
	w_{\alpha,\beta}(S) = \begin{cases}
		v(S) & S \in \K,\\
		\sum_{j \in S}\left(v(j) + \alpha_j\Delta\right) & S \notin \K \text{ and } \lvert S \rvert =n-1.\\
		-\beta_S + \sum_{j \in S}\left(v(j) + \alpha_j\Delta\right) & \text{otherwise.}.\\
	\end{cases}
\end{equation}

Although the sets of $\P$-extension, $\Co$-extensions and $\S$-extensions are bounded~\cite{Masuya2016}, the set of $\oneC$-extensions is clearly unbounded. Regarding their relations, in general, none of the sets is a subset of the set of $\oneC$-extensions though. However, we can say something about the relation between $\Co$-extensions and $\oneC$-extensions.
\begin{lemma}\label{lem:walfa-is-convex}
	For $\Delta \geq 0$, games $(N,\walfa)$ defined in~\eqref{eq:oneC-walfa} are convex.
\end{lemma}
\begin{proof}
	Follows from Lemma~\ref{lem:extreme-games}, because $(N,v^k)$ for $k \in N$ are convex and from the closedness of the class of convex games on convex combinations, because we have $\walfa = \sum_{k \in N}\alpha_iv^k$.
\end{proof}
Lemma~\ref{lem:walfa-is-convex} means that a minimal incomplete game $\pgame$ is $\Co$-extendable if and only if it is $\oneC$-extendable and the set of all games $(N,\walfa)$ is a subset of the set of $\Co$-extensions, i.e. 
\begin{equation}\label{eq:walfa-subset-convex-extensions}
	\bigcup\limits_{\substack{\alpha \in \Rn_+\\ \alpha(N)=1}} \walfa \subseteq \Co(v).
\end{equation}

\section{Approximations of solution concepts}\label{sec4}
Imagine that we want to determine a solution $\mathcal{S}(v)$ of a cooperative game $(N,v)$. However, all we know about this game is represented by an incomplete game $(N,\K,v)$ and the knowledge that the complete game is from a class $C \subseteq \Gamma^n$. The set of $C$-extensions represents the set of possible candidates with $(N,v)$ among them. This means, that by computing $\mathcal{S}(w)$ for every $C$-extension $(N,w)$, we also compute $\mathcal{S}(v)$. The only problem is that we cannot distinguish which one of the $C$-extensions is $(N,v)$. However, by considering the union of $\mathcal{S}(w)$ for all $C$-extensions $(N,w)$, we are sure that $\mathcal{S}(v)$ is contained in this set. Similarly, if we consider the intersection of all solutions $\mathcal{S}(w)$, we have a set of payoff vectors which are guaranteed to be a subset of $\mathcal{S}(v)$. The idea is formally captured by the following definition.
\begin{definition}\label{def:weak-solution-concept}
	Let $\pgame$ be a $C$-extendable incomplete game and $\mathcal{S}\colon C \to 2^{\Rn}$ a solution concept on $C \subseteq \Gamma^n$. Then by \emph{weak solution} $\cup\mathcal{S}(C,\K) \colon C(\K) \to 2^{\Rn}$, we mean
	\[
	\cup\mathcal{S}(C,\K)(v) \coloneqq \bigcup_{w \in C(v)}\mathcal{S}(w).
	\]
\end{definition}
We write $\cup\mathcal{S}(C)$ instead of $\cup\mathcal{S}(C,\K)$ whenever $\K$ is apparent from the context. In a similar way, we can define the \textit{strong solution} $\cap\mathcal{S}(C,\K)$ where the union is replaced by the intersection. It is clear from the definition that
\begin{equation}
	\cap\mathcal{S}(C,\K) \subseteq \mathcal{S}(v) \subseteq \cup\mathcal{S}(C,\K)
\end{equation}
and the difference between the sets depend heavily on both $C$ and $\mathcal{K}$. If for example $\K=2^N$, all three sets coincide. If on the contrary $\K=\{\emptyset\}$, for most of the standard solution concepts, the relations become $\emptyset \subseteq \mathcal{S}(v) \subseteq \mathbb{R}^n$. Also, the more restrictive $C$ is, the less $C$-extensions are considered, thus the stronger the approximations get. The ultimate goal is to find a compromise between information provided by $(C,\K)$ and the strength of the approximations.

We initiate the research by studying the weak and the strong solutions of minimal incomplete games. An advantage of this class of incomplete games is that the imputation set $\mathcal{I}(v)$ is determined as it depends only on values of singletons and the grand coalition. Moreover, it holds for every extension $(N,w)$ that $\mathcal{I}(w) = \mathcal{I}(v)$. Since many solution concepts are subsets of the imputation set (all but one studied in this paper), we get an initial upper bound $\cup\mathcal{S}(C)(v) \subseteq \mathcal{I}(v)$ for the weak solution. For the strong solution, an initial lower bound is given by the empty set, $\emptyset \subseteq \cap\mathcal{S}(v)$. The goal of this paper is to consider different sets of $C$-extensions and decide if the weak and the strong solutions lead to better approximations than the initial bounds. We consider standard classes of games such as the superadditive, convex, positive extensions but also 1-convex extensions as their analysis helps in understanding of the former ones.

\begin{table}
	\caption{Strong solutions $\cap\mathcal{S}(v)$}
	\begin{center}
		\begin{tabular}{ |c|c|c|c|c|c|c|c| } 
			\hline
			$\cap$ & $\phi$ & $\tau$ & $\eta$ & $K$ & $K^*$ & $C$ & $W$\\ 
			\hline
			$C_1^n$ & \texttimes & \texttimes & \texttimes & $\emptyset$ & $\emptyset$ & $\emptyset$ & $\emptyset$ \\
			\hline
			$S^n$ & \texttimes & \texttimes & \texttimes & $\emptyset$ & $\emptyset$ & $\emptyset$ & $\emptyset$ \\
			\hline
			$C^n$ & \texttimes & \texttimes & \texttimes & $\emptyset$ & $\emptyset$ & $\emptyset$ & $\emptyset$ \\
			\hline
			$P^n$ & \texttimes & \texttimes & \texttimes & $\emptyset$ & $\emptyset$ & $\emptyset$ & $\emptyset$ \\
			\hline
		\end{tabular}
	\end{center}
\end{table}

Unfortunately, the combination of minimal information together with the mentioned classes of games is too weak to provide any good lower approximation. This is because even for the most restrictive class of positive games, there are $P^n$-extensions with different single-point solutions which cannot be equal.

\begin{table}
	\caption{Weak solutions $\cup\mathcal{S}(v)$}
	\begin{center}
		\begin{tabular}{ |c|c|c|c|c|c|c|c| } 
			\hline
			$\cup$ & $\phi$ & $\tau$ & $\eta$ & $K$ & $K^*$ & $C$ & $W$ \\ 
			\hline
			$C_1^n$ & $\supsetneq$ & $=$ & $=$ & $=$ & $\supseteq$ & $=$ & $\supsetneq$ \\
			\hline
			$S^n$ & $\supseteq$ & $=$ & $=$ & $=$ & $=$ & $=$ & $\supseteq$ \\
			\hline
			$C^n$ & $=$ & $=$ & $=$ & $=$ & $=$ & $=$ & $=$ \\
			\hline
			$P^n$ & $\subsetneq$ & $\subsetneq$ & $\subsetneq$  & $\subsetneq$ & $\subsetneq$ & $=$ & $=$ \\
			\hline
		\end{tabular}
		\footnotetext{The relation between the weak solution concepts and the imputation set.}
	\end{center}
\end{table}

Regarding the weak solutions, the situation is much more interesting. For the set of 1-convex extensions, the solutions are either supersets or equal to the imputation set, thus we do not yield any approximations. Superadditive extensions are still too general to provide interesting upper approximations, however for the class of convex extensions, the weak Shapley value is already equal to $\mathcal{I}(v)$. For the class of positive extensions, nontrivial approximations are given for all solution concepts but the core and the Weber set (which coincide on this class of cooperative games).

In the rest of this section, we provide formal proofs for all of our claims.

\subsection{The Shapley value}
First a technical lemma, which expresses the Shapley value of every $\oneC$-ex\-ten\-sion.
\begin{lemma}\label{lem:shapley-oneC}
	Let $\pgame$ be a $\oneC$-extendable minimal incomplete game. It holds for every $\oneC$-extension $(N,\walfabeta)$ that
	\begin{equation}
		\footnotesize{\phi_i(\walfabeta) = v(i) + \alpha_i \Delta + \frac{1}{n}\left(\sum_{T \in E: i \notin T}\beta_T{n-1 \choose t}^{-1} - \sum_{T \in E:i \in T}\beta_T{n-1 \choose t-1}^{-1}\right)}.
	\end{equation}
\end{lemma}

\begin{proof}
	From linearity of the Shapley value and the definition of $\walfabeta$~\eqref{eq:oneC-walfabeta}, it holds $\phi(\walfabeta) = \sum_{i \in N}\alpha_i\phi(v^k) + \sum_{T \in E}\beta_T\phi(e_T)$. We compute both $\phi(v^k)$ and $\phi(e_T)$ using the definition of the Shapley value. For $S \subseteq N \setminus i$, it holds 
	\begin{equation}\label{eq:vertex-marginal-value}
		v^k(S \cup i) - v^k(S) = \begin{cases}v(i) + \Delta & \text{if } k=i,\\
			v(i) & \text{if } k\neq i.\\
		\end{cases}
	\end{equation}
	Notice,~\eqref{eq:vertex-marginal-value} is dependent only on $i$ and not on $S$. Further, for $X \in \{v(i),v(i) + \Delta\}$, it holds that
	\begin{equation}\label{eq:distribute-X}
		\frac{1}{n}\sum_{S \subseteq N \setminus i} {n-1 \choose s}^{-1} X = X\frac{1}{n}\sum_{S \subseteq N \setminus i} {n-1 \choose s}^{-1}.
	\end{equation}
	Modifying the sum is an easy exercise using the following identity:
	\begin{equation}\label{eq:mere-exercise}
		\sum_{S \subseteq N \setminus i} {n-1 \choose s}^{-1} = \sum_{j=0}^{n-1}{n-1 \choose j}{n-1 \choose j}^{n-1} = n.
	\end{equation}
	Combining together~\eqref{eq:vertex-marginal-value},~\eqref{eq:distribute-X} and~\eqref{eq:mere-exercise} in the expression of $\phi_i(v^k)$ from Definition~\ref{def:shapley}, it follows
	\[
	\phi_i(v^k) = \begin{cases}
		v(i) + \Delta & \text{if } i=k,\\
		v(i) & \text{if } i \neq k.\\
	\end{cases}
	\]
	In a similar manner, for game $(N,e_T)$, we can derive 
	\begin{equation}\label{eq:ray-marginal-value}
		e_T(S \cup i) - e_T(S) = \begin{cases}
			-1 & \text{if } S \cup i = T,\\
			1 & \text{if } S = T,\\
			0 & \text{otherwise.}
		\end{cases}
	\end{equation}
	Combining~\eqref{eq:ray-marginal-value} with Definition~\ref{def:shapley}, it follows
	\[
	\phi_i(e_T) = \begin{cases}
		-\frac{1}{n}{n-1 \choose t-1}^{-1} \text{if } i \in T,\\
		\frac{1}{n}{n-1 \choose t}^{-1} \text{if } i \notin T.\\
	\end{cases}
	\]
	This concludes the proof.
\end{proof}

Another lemma is about the boundedness of $\cup\phi(\oneC)$.

\begin{lemma}\label{lem:shapley-unbounded}
	The weak Shapley value $\cup\phi(\oneC)$ is unbounded.
\end{lemma}
\begin{proof}
	For every lower bound $b \in \mathbb{R}$, there is $(N,\walfabeta)$ with $\alpha_i=0$ and $\beta_T = 0$ for every $T \in E$ except for exactly one $S \in E$ such that $i \in S$, for which $\beta_S = {n-1 \choose t-1}b + \epsilon$ for $\epsilon > 0$. It holds $\phi_i(\walfabeta) < b$. 
\end{proof}

Results from Lemma~\ref{lem:shapley-oneC} and~\ref{lem:shapley-unbounded}, combined with Lemma~\ref{lem:walfa-is-convex} on the relation between $\oneC$-extensions and $\Co$-extensions, implicate the following.

\begin{theorem}\label{thm:shapley-oneC-and-convex}
	Let $\pgame$ be minimal incomplete game. it holds
	\begin{enumerate}
		\item $\cup\phi(\oneC)\supsetneq \mathcal{I}(v)$ if $\pgame$ is $\oneC$-extendable,
		\item $\cup\phi(\Co) = \mathcal{I}(v)$ if $\pgame$ is $\Co$-extendable,
		\item $\cup\phi(\S) \supseteq \mathcal{I}(v)$ if $\pgame$ is $\S$-extendable.
	\end{enumerate}
\end{theorem}

\begin{proof}
	From Lemma~\ref{lem:shapley-oneC}, it follows
	\begin{equation}
		\bigcup\limits_{w_\alpha \in \oneCv}\phi(w_\alpha) = I(v).
	\end{equation}
	Together with Lemma~\ref{lem:shapley-unbounded}, we have $\cup\phi(\oneC)\supsetneq \mathcal{I}(v)$. Further, from Lemma~\ref{lem:walfa-is-convex}, we have
	\[\bigcup\limits_{w_\alpha \in \oneCv}w_\alpha \subseteq \Co(v),\]
	thus $\mathcal{I}(v) \subseteq \cup\phi(\Co) \subseteq \mathcal{I}(v)$ holds. Finally, since $\Co$-extensions form a subset of $\S$-extensions, we can deduce that $\cup\phi(\S)\supseteq \mathcal{I}(v)$.
\end{proof}

In general, there are superadditive games for which the Shapley value is not from the imputation set, therefore we believe that the strict inclusion might hold for Theorem~\ref{thm:shapley-oneC-and-convex}.3. Finally, we turn our attention to the set of $\P$-extensions, which gives the best approximation. 

\begin{lemma}\label{lem:shapley-positive}
	Let $\pgame$ be a $\P$-extandable minimal incomplete game. It holds for every $\P$-extension $(N,\wa)$ that
	\begin{equation}
		\phi_i(\wa) = v(i) + \Delta\sum_{T \in N_1, i \in T}\frac{\alpha_T}{\lvert T \rvert}.
	\end{equation}
\end{lemma}
\begin{proof}
	Every $\P$-extensions $\wa$ can be expressed as 
	\[
	w = \sum_{i \in N}v(i)u_{i} + \Delta\sum_{T \in N_1}\alpha_Tu_T,
	\]
	where $(N,u_T)$ are unanimity games. Now from the linearity of the Shapley value,
	\begin{equation}\label{eq:shapley-positive}
		\phi(w) = \sum_{i \in N}v(i)\phi(u_{i}) + \Delta\sum_{T \in E_1}\alpha_T\phi(u_T).
	\end{equation}
	Combining~\eqref{eq:shapley-positive} and Lemma~\ref{lem:shapley-unanimity} concludes the proof.
\end{proof}

To analyse the relation between $\mathcal{I}(v)$ and $\cup\phi(\P)$, we show that there are imputations which cannot be the Shapley value of any $\P$-extension.

\begin{theorem}\label{thm:shapley-positive}
	Let $\pgame$ be a $\P$-extandable minimal incomplete game. It holds
	\[
	\cup\phi(\P) \subsetneq \mathcal{I}(v).
	\]
\end{theorem}
\begin{proof}
	The inclusion $\cup\phi(\P) \subseteq \mathcal{I}(v)$ is immediate from $\cup\phi(\P)\subseteq \cup\phi(\Co)$ and Theorem~\ref{thm:shapley-oneC-and-convex}. To prove that the strict inclusion holds, we show that $I^k \notin \cup\phi(\P)$ for every $k \in N$. For a contradition, suppose $I^k \in \cup\phi(\P)$.  It means, there is a $\P$-ex\-ten\-sion $(N,\wa)$ such that $\phi(\wa)=I^k$. From $\phi_i(\wa) = I^k_i$, we get
	\begin{equation}\label{eq:positive-shapley-sum}
		\sum_{T \in N_1, i \in T}\frac{\alpha_T}{\lvert T \rvert} = 0 \text{ for every }i \neq k.
	\end{equation}
	Since $\alpha_T \geq 0$, each $\alpha_T$ in each of the sums from~\eqref{eq:positive-shapley-sum} is actually equal to zero. But this means that $\alpha_T=0$ holds for every $T \in N_1$, which leads to a contradiction, because
	\[
	0 = \sum_{T \in N_1} \alpha_T = 1.
	\]
\end{proof}

\subsection{The $\tau$-value}
The analysis of the weak $\tau$-value can be done is a similar manner as the analysis of the weak Shapley value. We begin by an explicit description of the $\tau$-value of every $\oneC$-extension. Interestingly, even tough the set of $\oneC$-extensions forms an unbounded convex cone, the weak $\tau$-value is a bounded set.
\begin{lemma}\label{lem:oneC-tau-value}
	Let $\pgame$ be $\oneC$-extendable minimal incomplete game. It holds for every $\oneC$-extension $(N,\walfabeta)$ that
	\[
	\tau(\walfabeta) = I^\alpha.
	\]
\end{lemma}
\begin{proof}
	The assertion follows immediately from Theorem~\ref{thm:1convex-tau}, if we show that $g^\walfabeta(N)=0$. From Lemma~\ref{lem:linear-combination}, it holds
	\[
	g^\walfabeta(N) = \sum_{i \in N}\alpha_i g^{v^i}(N) + \sum_{T \in E}\beta_T g^{e_T}(N).
	\]
	For every $i \in N$, it holds $b^{v^i}(N)=\sum_{j \in N}v(j) + \Delta = v(N)$, thus $g^{v^i}(N) = b^{v^i}(N) - v^i(N) = 0$. Similarly, for every $T \in E$, $g^{e_T}(N)=0$. We conclude $g^{\walfabeta}(N)=0$.
\end{proof}
\begin{theorem}
	Let $\pgame$ be minimal incomplete game. It holds
	\begin{enumerate}
		\item $\cup\tau(\oneC) = \I(v)$ if $\pgame$ is $\oneC$-extendable,
		\item $\cup\tau(\Co) = \I(v)$ if $\pgame$ is $\Co$-extendable,
		\item $\cup\tau(\S) = \I(v)$ if $\pgame$ is $\S$-extendable.
	\end{enumerate}
\end{theorem}
\begin{proof}
	The first assertion follows from the fact that 
	\[
	\bigcup\limits_{\walfabeta \in \oneCv}\tau(\walfabeta) = \bigcup\limits_{\substack{\alpha \in \mathbb{R}^n:\\ \alpha(N)=1,\\\forall i:\alpha_i \geq 0}} I^\alpha = \I(v).
	\]
	The second follows from Lemma~\ref{lem:walfa-is-convex} and the fact that 
	\[
	\I(v) = \bigcup\limits_{\walfa \in \oneCv}\tau(\walfa) \subseteq \cup\tau(\Co) \subseteq \I(v)
	\]
	and the final assertion from $\mathcal{I}(v) = \cup\tau(\Co) \subseteq \cup\tau(\S) \subseteq \I(v)$.
\end{proof}
The weak $\tau$-value do not yield non trivial approximation for classes of $\oneC$-ex\-ten\-sions, $\S$-extensions and $\Co$-extensions. This changes when we further restrict to positivity.
\begin{lemma}
	Let $\pgame$ be $\P$-extendable minimal incomplete game. It holds for every $\P$-extension $(N,\wa)$ that
	\[
	\tau_i(\wa) = v(i) + \Delta\frac{\sum\limits_{T \in N_1: i \in T}\alpha_T}{\sum\limits_{\tiny{T \in N_1}}\alpha_T \lvert T \rvert}.
	\]
	
\end{lemma}
\begin{proof}
	From Theorem~\ref{thm:1convex-tau}, it holds for every $\P$-extension $(N,\wa)$ that
	\[
	\tau_i(\wa) = b_i^\wa - \frac{g^\wa(N)}{\sum_{i \in N}g^\wa(i)}g^\wa(i).
	\]
	From linearity of upper vectors and the gap function (Lemma~\ref{lem:linear-combination}), we can express
	\[
	g^\wa(N) = \sum_{T \in N_1}\alpha_Tg^{v_T}(N) \text{ and } b^\wa(N) = \sum_{T \in N_1}b^{v_T}(N).
	\]
	From the definition of the upper vector, we have
	\[
	b_i^{v_T} = \begin{cases}
		v(i) + \Delta & \text{if } i \in T,\\
		v(i) & \text{if } i \notin T.\\
	\end{cases}
	\]
	Further, it is easy to express
	\begin{itemize}
		\item $b^{v_T}(N) = \Delta \lvert T \rvert  + \sum_{i \in N}v(i)$,
		\item $b^{\wa}(N) = \Delta\sum_{T \in N_1} \lvert T \rvert + \sum_{i \in N}v(i)$,
		\item $g^{\wa}(N) = \Delta \left(\sum_{T \in N_1} \lvert T \rvert - 1\right)$,
		\item $b_i^{\wa} = v(i) + \Delta \sum_{T \in N_1, i \in T}\alpha_T$,
		\item $g^{\wa}(i) = \Delta \sum_{T \in N_1, i \in T}\alpha_T$.
	\end{itemize}
	Combining these expressions concludes the proof.
\end{proof}

\begin{theorem}
	Let $\pgame$ be a $\P$-extendable minimal incomplete game. It holds
	\[
	\cup\tau(\P)\subsetneq \I(v).
	\]
\end{theorem}
\begin{proof}
	Similarly to the proof of Theorem~\ref{thm:shapley-positive}, we show $I^k \notin \cup\tau(\P)$ for any $k \in N$. For a contradition, if $I^k \in \cup\tau(\P)$, then there is a $\P$-extension $(N,\wa)$ such that $\tau(\wa)=I^k$. It follows from $\tau_i(w) = I^k_i$ that
	\begin{equation}\label{eq:tau-sum}
		\sum_{T \subseteq N, i \in T}\alpha_T = 0 \text{ for every }i \neq k.
	\end{equation}
	As in the proof of Theorem~\ref{thm:shapley-positive}, from $\alpha_T \geq 0$ and~\eqref{eq:tau-sum}, it follows $\alpha_T=0$ for every $T \in N_1$, thus
	\[
	0 = \sum_{T \subseteq N_1} \alpha_T = 1.
	\]
\end{proof}

\subsection{The nucleolus and the (pre-)kernel}
Although different in general, for the class of convex games, the nucleolus, the prekernel, and the kernel coincide. The prekernel and the kernel even coincide for the class of superadditive games and for both superadditive and convex games, these three solution concepts are subsets of $\I(v)$. We show that for superadditive and convex extensions, the approximations do not yield any strengthening over $\emptyset \subseteq \mathcal{S}(v) \subseteq \I(v)$. For positive extensions, however, we are able to show that there is an improvement for the weak solution.

\begin{lemma}\label{lem:nucleolus-oneC}
	Let $\pgame$ be a $\oneC$-extendable minimal incomplete game. It holds for every  $\oneC$-extension $(N,\walfabeta)$ that
	\[
	\eta^*(\walfabeta) = I^\alpha.
	\]
\end{lemma}
\begin{proof}
	From Theorem~\ref{thm:oneC-nucleolus-equals-tau} and Lemma~\ref{lem:oneC-tau-value}, we have $\eta(\walfabeta)=\tau(\walfabeta)=I^\alpha$.
\end{proof}

\begin{theorem}
	Let $\pgame$ be a minimal incomplete game. It holds
	\begin{enumerate}
		\item $\cup\eta(\S) = \cup\K(\S)=\cup\K^*(\S)=\mathcal{I}(v)$ if $\pgame$ is $\S$-extendable,
		\item $\cup\eta(\Co) = \cup\K(\Co)=\cup\K^*(\Co)=\mathcal{I}(v)$ if $\pgame$ is $\Co$-extendable,
		\item $\cup\eta(\oneC) = \cup\K(\oneC) = \I(v)$ if $\pgame$ is $\oneC$-extendable,
		\item $\cup\K^*(\oneC) \supseteq \I(v)$ if $\pgame$ is $\oneC$-extendable.
	\end{enumerate}
\end{theorem}
\begin{proof}
	The theorem follows immediately from Lemma~\ref{lem:nucleolus-oneC}, the relation between the sets of $\oneC$-extensions, $\Co$-extensions and $\S$-extensions (Lemma~\eqref{lem:walfa-is-convex} and Equation~\eqref{eq:relation-of-extensions}) and the relations between $\eta$, $\K$ and $\K^*$ for convex games (Theorem~\ref{thm:convex-nucleolus-coincidence}).
\end{proof}

As the coincidence of the prekernel and kernel with the nucleolus suggests, from Lemma~\ref{lem:nucleolus-oneC}, it is clear that the strong solution yields an empty set for all three sets of $\oneC$-extensions, $\Co$-extensions and $\S$-extensions. 
Also, the weak nucleolus is not a non trivial approximation even when convexity is considered. Yet again, this changes when restricting to positivity. However, the positivity is still not strong enough for a non trivial lower approximation.

\begin{theorem}
	Let $\pgame$ be a $\P$-extendable minimal incomplete game. It holds
	\[
	\cup\K^*(\P) \subsetneq \I(v).
	\]
\end{theorem}
\begin{proof}
	We show that $\I^k \notin \cup\K^*(\P)$ for every $k \in N$. For a contradiction, let $(N,\wa) \in \P(v)$ such that $\K^*(\wa)=\{I^k\}$. The prekernel is single-valued because by Theorem~\ref{thm:convex-nucleolus-coincidence}, $\K^*(\wa)=\eta(\wa)$. We can express the excess of every $S \subseteq N$ as
	\[
	e(S,I^k,\wa) = \begin{cases}
		\Delta \left(\sum\limits_{T \in N_1, T \subseteq S}\alpha_T - 1\right) & \text{if } k \in S,\\
		\Delta\sum\limits_{T \in N_1, T \subseteq S}\alpha_T & \text{if } k \notin S.\\
	\end{cases}
	\]
	It holds that
	\[e(S,I^k,\wa) \leq e(S\cup i,I^k,\wa)\text{ for }i \neq k\]
	and 
	\[e(S,I^k,\wa) \leq 0\text{ if }k \in S.\] 
	Therefore, the maximal surplus $s_{ij}(I^k,\wa)$ is attained for $N \setminus \{j,k\}$ if $i \neq k$ and for $N \setminus \{j\}$ if $i=k$. Consider the relation $s_{ki}(I^k,\wa)=s_{ik}(I^k,\wa)$. It can be equivalently expressed as
	\[
	\Delta \left(\sum\limits_{T \in N_1, T \subseteq N \setminus i}\alpha_T - 1\right) = \Delta\sum\limits_{T \in N_1, T \subseteq N \setminus k}\alpha_T,
	\]
	or
	\[
	\sum\limits_{T \in N_1, T \subseteq N \setminus i}\alpha_T - \sum\limits_{T \in N_1, T \subseteq N \setminus k}\alpha_T = 1.
	\]
	From $\alpha_T \geq 0$ for $T \in N_1$ and $\sum_{T \in N_1}\alpha_T=1$, it holds for every $T \in N_1, T \subseteq N \setminus k$ that $\alpha_T=0$. This means that
	\[
	\sum\limits_{T \in N_1, T \subseteq N \setminus i}\alpha_T = 1 \text{ for } \forall i \neq k.
	\]
	To satisfy all these conditions, there has to be a coalition $S$ with $\alpha_S = 1$ from \[\bigcap\limits_{i \neq k} \{T \in N_1 \mid T \subseteq N \setminus i\} = \emptyset.\]
	This is a contradiction. 
\end{proof}

To show that the strong nucleolus is empty, one only has to show there are two $\P$-extensions with different nucleolus. This is trivial and left to the reader.

\subsection{The core and the Weber set}

The core is always a selection from the imputation set $\I(v)$. We show that even for the set of $\P$-extensions, the weak solution is equal to the imputation set. However, there is a $\P$-extension $(N,\wa)$ for which $\Core(\wa) = \I(v)$, therefore without further knowledge, it might be that actually the approximation is equal to the solution.

\begin{theorem}\label{thm:weak-core-positive}
	Let $\pgame$ be a minimal incomplete game. It holds
	\begin{enumerate}
		\item $\cup\Core(\P) = \cup\W(\P)=\I(v)$ if $\pgame$ is $\P$-extendable,
		\item $\cup\Core(\Co) = \cup\W(\Co) = \I(v)$ if $\pgame$ is $\Co$-extendable,
		\item $\cup\Core(\S) = \cup\Core(\S) = \I(v)$ if $\pgame$ is $\S$-extendable.
	\end{enumerate}
\end{theorem}
\begin{proof}
	One of the vertices of the set of $\P$-extensions, $(N,v_N)$, defined in~\eqref{eq:def-vt}, satisfies $\W(v_N)=\I(v)$. The rest follows from the relation between the core and the Weber set (Theorem~\ref{thm:convex-weber-equals-core}) and the relation between sets of $C$-extensions from ~\eqref{eq:relation-of-extensions}.
	To show that $\W(v_N)=\I(v)$, we actually show that the vertices of these sets coincide. Consider $\sigma \in \Sigma_n$ and $S_{\sigma(1)},\dots, S_{\sigma(n)}$. For the correspoing marginal vector $m_\sigma^{v_N}$ it holds
	\[
	(m_\sigma^{v_N})_i = v_N(S_{\sigma(i)} \cup i) - v_N(S_{\sigma(i)}) = \begin{cases}
		v(i) + \Delta & \text{if $\sigma(i)=n$,}\\
		v(i) & \text{otherwise.} \\
	\end{cases}
	\]
	It follows $m_{\sigma}^{v_N} = I^k$ where $\sigma(k)=n$, which concludes the proof.
\end{proof}

Game $(N,v_N)$ is not 1-convex, therefore the previous result cannot be applied for $\oneC$-extensions. We also have to be more cautious in our analysis, because it does no longer hold that $\W(\walfabeta)=\Core(\walfabeta)$ because for $\beta \neq 0$, extension $(N,\walfabeta)$ might not be convex anymore. However, we are able to show much more for the set of $\oneC$-extensions.

\begin{lemma}\label{lem:core-oneC}
	Let $\pgame$ be $\oneC$-extendable incomplete game. For every $\oneC$-extensions $(N,\walfabeta)$ it holds
	\[
	\Core(\walfabeta)=\{I^\alpha\}.
	\]
\end{lemma}
\begin{proof}
	If we show that $g^\walfabeta(N)=0$, it follows from Theorem~\ref{thm:1convex-core}, that $\Core(\walfabeta) = \{b^\walfabeta\}$. Further, for every $i \in N$,
	\[
	b^\walfabeta_i = \walfabeta(N) - \walfabeta(N \setminus i) = v(i) + \alpha_i\Delta = I^\alpha_i,
	\]
	which concludes the proof. To prove that $g^\walfabeta(N)=0$, we employ Lemma~\ref{lem:linear-combination}, to express
	\[
	g^\walfabeta(N) = \sum_{i \in N}\alpha_i g^{v^i}(N) + \sum_{T \in E}\beta_T g^{e_T}(N).
	\]
	It follows from the definition of the gap function and definitions of games $(N,v^i)$ and $(N,e_T)$ that both
	\[g^{v^i}(N) = 0\text{ and }g^{e_T}(N) = 0.\]
\end{proof}

\begin{theorem}
	Let $\pgame$ be a $\oneC$-extendable minimal incomplete games. It holds that
	\[
	\cup\Core(\oneC)=\mathcal{I}(v).
	\]
\end{theorem}
\begin{proof}
	Follows immediatelly from Lemma~\ref{lem:core-oneC}.
\end{proof}

To analyse the weak Weber set, we can yet again explicitly express the Weber set of every $\oneC$-extension.

\begin{lemma}\label{lem:weber-oneC}
	Let $\pgame$ be a $\oneC$-extendable minimal incomplete game. for every $\oneC$-extension $(N,\walfabeta)$ it holds $\W(\walfa) = conv\{\marg{\walfabeta} \mid \sigma \in \Sigma_n\}$, where
	\[
	\margi{\walfabeta}{i} = \begin{cases}
		v(i) & \sigma(1)=i \\
		v(i) + (\alpha_i + \alpha_j)\Delta - \beta_{S_{\sigma(j)}\cup \{j\}} & \sigma(1)=j, \sigma(2) = i\\
		v(i) + \alpha_i\Delta + \beta_{S_{\sigma(i)}} & \sigma(n-1) = i \\
		v(i) + \alpha_i\Delta & \sigma(n)=i \\
		v(i) + \alpha_i\Delta + (\beta_{S_{\sigma(i)}} - \beta_{S_{\sigma(i)} \cup \{i\}})
	\end{cases}
	\]
\end{lemma}
\begin{proof}
	Follows from the definition of marginal vectors, defined in~\eqref{eq:marg-vector-definition} and the definition of $(N,\walfabeta)$, defined in~\eqref{eq:oneC-walfabeta}.
\end{proof}

\begin{theorem}
	Let $\pgame$ be $\oneC$-extendable incomplete game. It holds
	\begin{equation*}
		\I(v) \subsetneq \cup\W(\oneC).
	\end{equation*}
\end{theorem}
\begin{proof}
	From Lemma~\ref{lem:weber-oneC}, we know 
	\[
	\I(v) = \bigcup\limits_{\walfa \in \oneC(v)}\Core(\walfa)=\bigcup\limits_{\walfa \in \oneC(v)}\W(\walfa) \subseteq \cup\W(\oneC).
	\] 
	The strict inclusion can be derived from Lemma~\ref{lem:weber-oneC}, when considering $(m_\sigma^\walfabeta)_i$ with $i \neq \sigma(1)$ such that 
	\[
	(m_\sigma^\walfabeta)_i = v(i) + \alpha_i\Delta + (\beta_{S_{\sigma(i)}} - \beta_{S_{\sigma(i)} \cup \{i\}}).
	\]
	Fixing other $\beta_T$, with $\beta_{S_{\sigma(i)}}$ going to infinity, $(m_\sigma^\walfabeta)_i$ tends to infinity and with $\beta_{S_{\sigma(i)} \cup \{i\}}$ going to infinity, $(m_\sigma^\walfabeta)_i$ tends to minus infinity.
\end{proof}

From Lemma~\ref{lem:core-oneC}, it is easy to see that the strong core of the set of $\oneC$-extensions (thus also $\Co$-extensions and $\Sv$-extensions) is an empty set. The Lemma does not apply to the set of $\P$-extensions, because none of games $(N,\walfa)$ is positive. To prove the emptyness of the strong core of $\P$-extensions, we show that for every $I^\alpha \in \I(v)$, there is $(N,\wa)$ such that $I^\alpha \notin \Core(\wa)$.

\begin{theorem}\label{thm:strong-core-positive}
	Let $\pgame$ be $\P$-extendable minimal game. It holds
	\[
	\cap\Core(\P) = \cap\W(\P)=\emptyset.
	\]
\end{theorem}
\begin{proof}
	Let $I^\alpha \in \I(v)$ and distinguish two cases. If $\alpha_1 + \alpha_2 = 1$ then choose $(N,\wa)$ such that for $i,j$ different from $1,2$, we have $\alpha_{\{i,j\}} = 1$ and $\alpha_T=0$ otherwise. It follows that
	\[
	I^k(\{i,j\}) = v(i) + v(j) < v(i) + v(j) + \Delta = \wa(\{i,j\}),
	\]
	therefore $I^\alpha \notin \Core(\wa)$ because the coalitional rationality for $\{i,j\}$ is not satisfied. 
	
	Similarly, if $\alpha_1 + \alpha_2 < 1$, choose $(N,\wa)$ such that $\alpha_{\{1,2\}}=1$ and $\alpha_T = 0$ otherwise. Again,
	\[
	I^k(\{1,2\}) = v(1) + v(2) + (\alpha_1+\alpha_2)\Delta < v(1) + v(2) + \Delta = \wa(\{1,2\}),
	\]
	meaning the coalitional rationality is not satisfied for $\{1,2\}$, thus $I^\alpha \notin \Core(\wa)$.
\end{proof}

\section{Conclusion}\label{sec5}
We introduced a new method to approximate the solution concepts of cooperative games when only partial information is given. The method was demonstrated on a class of minimal incomplete cooperative games. We investigated different solution concepts and showed how the approximations depend on the knowledge about the underlying game, specifically on the class of the underlying game. We note that incomplete games with different structure of $\K$ can be analysed in a similar manner. Interesting candidates might be e.g. 
\begin{enumerate}
	\item $\K = \{T \subseteq N \mid i \in T\}$,
	\item $\K = \{T \subseteq N \mid \lvert T \rvert \leq k\}$,
	\item $\K=\{N,\emptyset\} \cup \{N \setminus i \mid i \in N\}$ (for motivation, see~\cite{Bok2020}), 
	\item classes of games given in~\cite{Grabisch2011}. 
\end{enumerate}

We showed that for minimal incomplete games, most of the weak solutions become non trivial (not equal to the imputation set) when we switch between the sets of $\Co$-extensions and $\P$-extensions. The first thing we would like to understand in the future is when exactly this change happens. A possible approach might be to study the weak solutions for sets of \emph{$k$-monotonic extensions} (see~\cite{Grabisch2016}). In this view, convex games are $2$-monotonic and positive games are $(2^n-2)$-monotonic. It further holds that $k$-monotonic games form a subset of $(k-1)$-monotonic games. 

Among other solution concepts, the core stands out because using our method, it is the only solution concept which do not yield non trivial approximation even for $\P$-extensions. To understand the core and also to be able to approximate it, it is important to find sets of $C$-extensions for which we get non trivial approximations. A possible candidate might be to further restrict $\P$-extensions which are $k$-additive, and the most restrictive class of 2-additive games (see~\cite{Grabisch2016}). We note that the proof of Theorem~\ref{thm:strong-core-positive} shows that the strong core of 2-additive $\P$-extensions is still empty, however, it might be still a different case for the weak core.

The next big step regarding the approximations is to derive tools to further analyse the \emph{strength} of the approximations. So far, we considered only the inclusion as a measure that one approximation is better than another. For example, if there is a $C$-extensions for which the solution concept is equal exactly to the weak solution (see proof of Theorem~\ref{thm:weak-core-positive}), the approximation is clearly best possible we can get if we consider that the underlying game is from $C$. But consider weak solutions of one-point solution concepts. Clearly, \emph{smaller} the weak solution (in volume, range of values, ...), the better approximation we are getting. We also want to address these questions in near future.


\bmhead{Acknowledgments}

The author would like to thank Milan Hlad\'{i}k and Michel Grabisch for discussions regarding the paper. The author was supported by SVV--2020--260578, by the Charles University Grant Agency (GAUK 341721) and by the Czech Science Foundation (GA\v{C}R 22-11117S).

\bibliography{bibliography}


\end{document}